\documentclass[11pt,oneside]{article}

\usepackage[hidelinks]{hyperref}

\usepackage[T1]{fontenc}
\usepackage{mathpazo}
\usepackage[margin=1in,includefoot]{geometry}

\usepackage{graphicx}

\usepackage{amsmath}
\usepackage{amssymb}
\DeclareMathOperator*{\argmax}{arg\,max}

\DeclareMathOperator{\Cov}{Cov}
\DeclareMathOperator{\E}{E}
\DeclareMathOperator{\Var}{Var}
\newcommand{\parfrac}[2]{\frac{\partial #1}{\partial #2}}
\newcommand{\R}{\mathbb{R}}

\usepackage{amsthm}
\newtheorem{corollary}{Corollary}
\newtheorem{lemma}{Lemma}
\newtheorem{proposition}{Proposition}

\usepackage{natbib}
\newcommand\citeapos[1]{\citeauthor{#1}'s (\citeyear{#1})}

\usepackage{multirow}
\usepackage{subcaption}

\usepackage{xcolor}

\title{Why do experts give simple advice?}
\author{
Benjamin Davies\thanks{
Department of Economics, Stanford University; bldavies@stanford.edu.
I thank Yunus Aybas, Doug Bernheim, Steve Callander, Matt Jackson, Spencer Pantoja, and Anirudh Sankar for helpful discussions.
}
}
\date{Draft version: September 2022}

\begin{document}

\maketitle

\begin{abstract}
    \noindent
    An expert tells an advisee whether to take an action that may be good or bad.
    He may provide a condition under which to take the action.
    This condition predicts whether the action is good if and only if the expert is competent.
    Providing the condition exposes the expert to reputational risk by allowing the advisee to learn about his competence.
    He trades off the accuracy benefit and reputational risk induced by providing the condition.
    He prefers not to provide it---i.e., to give ``simple advice''---when his payoff is sufficiently concave in the posterior belief about his competence.

    \vskip\baselineskip

    \noindent
    {\bfseries JEL codes}:
    D82, D83

    \noindent
    {\bfseries Keywords}:
    advice, career concerns, experts, information disclosure, reputation, science communication
\end{abstract}

\clearpage

\section{Introduction}

This paper discusses how reputational concerns impact the type of advice experts provide.
I frame advice as a ``rule'' or ``model'' that maps conditions to actions.
It makes a prediction of the \emph{ex post} ``correct'' action based on \emph{ex ante} unobserved conditions.
For example, it tells a patient whether to pursue medical treatment if they develop specific symptoms.
I assume that
\begin{enumerate}

    \item[(i)]
    the \emph{ex post}  correct action cannot be predicted perfectly, and

    \item[(ii)]
    the advisee does not know what conditions are relevant if the expert does not include them in their advice.

\end{enumerate}
These assumptions reflect the situation faced by public health officials and policy-makers early in the COVID-19 pandemic.
Assumption~(i) reflects how the mapping between medical conditions (e.g., being immunocompromised) and optimal policies (e.g., getting vaccinated) was uncertain because the relevant science was unsettled.
It implies that advice can be ``wrong'' (i.e., recommend the \emph{ex post} incorrect action) even if it includes all relevant conditions.
Assumption~(ii) reflects how the novelty of the relevant science prevented advisees from comparing experts' advice with public knowledge.
It implies that \emph{if} advice is wrong then experts cannot be punished for excluding relevant conditions.
Intuitively, advisees cannot observe bad outcomes and attribute them to causes experts ``should have known about'' because advisees do not know the set of possible causes.%
\footnote{
This became less true as the pandemic progressed and people started sharing (scientific and anecdotal) information about vaccination outcomes.
}

Under assumptions~(i) and~(ii), experts make strategic choices about the type of advice they provide.
Advice may be unconditional (e.g., ``do $x$'') or have conditions (e.g., ``do $x$ if $y$ happens and $z$ does not'').
Including conditions tailors the advice and makes experts appear more confident.
However, this confidence may hurt experts' reputations if their advice turns out to be wrong.
Consequently, experts may exclude relevant conditions---that is, ``simplify'' their advice---to hide their confidence and preserve their reputation.
I analyze this tension between providing better advice and managing reputational risk.

Section~\ref{sec:model-setup} embeds this tension in a simple model.
An expert (he) tells an advisee (she) whether to take an action that may be good or bad.
Both parties want the advisee to take the action if and only if it is good.
The expert may provide a binary condition under which to take the action.
This condition predicts whether the action is good if the expert is competent but not if he is incompetent.
Neither the advisee nor the expert know if he is competent.
He chooses between a ``simple rule,'' which says to take the action, and a ``complex rule,'' which says to take the action under the provided condition.
If he chooses the complex rule then the advisee learns about the expert's competence by whether the rule predicts the correct action.
If he chooses the simple rule then she learns nothing about the expert's competence.
The expert trades off the complex rule's superior predictive power with its ability to reveal information about his competence.

Section~\ref{sec:model-analysis} analyzes the expert's trade-off.
Choosing the complex rule over the simple rule induces a mean-preserving spread in the posterior belief about the expert's competence.
This spread makes him better off if his payoff is convex in the posterior and worse off if it is concave in the posterior.
In the convex case, the expert always chooses the complex rule because it is more accurate and induces an expected reputational benefit.
In the concave case, he chooses the complex rule if and only if its accuracy benefit exceeds its reputational cost.
Increasing the condition's predictive power makes the complex rule more accurate and more diagnostic of the expert's competence.
If his payoff is concave in the posterior then this increase in predictive power makes the expert more willing to choose the complex rule if and only if his reputational concerns are sufficiently weak.

Section~\ref{sec:wage-example} considers a setting in which the expert earns a wage if and only if the posterior belief about his competence exceeds a threshold.
If this threshold equals the prior belief, then the expert chooses the complex rule if and only if the wage is sufficiently low.
If the advisee wants to solicit the complex rule but cannot change the wage, then she can
(i)~lower the threshold and guarantee earning the wage, or
(ii)~raise the threshold and guarantee \emph{not} earning the wage.
Both~(i) and~(ii) remove the reputational risk that prevents the expert from choosing the complex rule.

Section~\ref{sec:extensions} considers some extensions of my analysis and Section~\ref{sec:literature} compares it to related literature.
Section~\ref{sec:conclusion} concludes.
Appendix~\ref{sec:proofs} contains proofs of my mathematical claims.


\section{Model}
\label{sec:model-setup}

\paragraph{Environment}

There are two players: an expert (he) and an advisee (she).

The expert tells the advisee whether to take an action.
The action is either good ($A=1$) or bad ($A=0$).
The expert and advisee know $\Pr(A=1)=0.5$.

The expert may provide a condition under which to take the action.
This condition predicts $A$ if the expert is competent ($\theta=1$) but not if he is incompetent ($\theta=0$).
The condition $X_\theta\in\{0,1\}$ has prevalence $\Pr(X_\theta=1)=x\in(0,1)$ and covariance $\Cov(A,X_\theta)=\theta\sigma$ with $\sigma\in(0,0.25]$.

The expert wants the advisee to take the action if and only if it is good.
He offers one of two rules $r\in\{S,C\}$:
\begin{enumerate}

    \item
    ``take the action'' (the ``simple'' rule $r=S$);

    \item
    ``take the action if and only if $X_\theta=1$'' (the ``complex'' rule $r=C$).

\end{enumerate}
Thus, he provides the condition if and only if he chooses the complex rule.
The advisee can verify the condition's prevalence and covariance if and only if the expert provides it.
If he provides it then she observes its value between receiving and implementing his advice.

Let $Y\in\{0,1\}$ indicate whether the advisee takes the ``correct'' action.
Then
\begin{align}
    \Pr(Y=1\,\vert\,r,\theta)
    &= \begin{cases}
        \Pr(A=1) & \text{if}\ r=S \\
        \Pr(A=X_\theta) & \text{if}\ r=C
    \end{cases} \notag \\
    &= \begin{cases}
        0.5 & \text{if}\ r=S \\
        0.5+2\theta\sigma & \text{if}\ r=C
    \end{cases} \label{eq:prob-Y-r-theta}
\end{align}
for each $\theta\in\{0,1\}$.

\paragraph{Beliefs}

Neither the advisee nor the expert know whether he is competent.
They have common prior $\Pr(\theta=1)=\pi_0\in(0,1)$.
Hence
\begin{align}
    \Pr(Y=1\,\vert\,r)
    &= \Pr(Y=1\,\vert\,r,\theta=1)\Pr(\theta=1)+\Pr(Y=1\,\vert\,r,\theta=0)\Pr(\theta=0) \notag \\
    &= \begin{cases}
        0.5 & \text{if}\ r=S \\
        0.5+2\sigma\pi_0 & \text{if}\ r=C.
    \end{cases} \label{eq:prob-Y-r}
\end{align}

The expert's ignorance of $\theta$ prevents him from trying to signal it.
Therefore, his chosen rule $r$ is independent of $\theta$.%
\footnote{
This remains true (in equilibrium) if the expert knows~$\theta$---see Section~\ref{sec:known-type}.
}

The players use~\eqref{eq:prob-Y-r-theta}, \eqref{eq:prob-Y-r}, and Bayes' rule to form a posterior belief
\begin{align}
    \Pr(\theta=1\,\vert\,r,Y)
    &= \frac{\Pr(Y\,\vert\,r,\theta=1)\Pr(\theta=1)}{\Pr(Y\,\vert\,r)} \notag \\
    &= \begin{cases}
        \pi_0 & \text{if}\ r=S \\
        \frac{(1+4\sigma)\pi_0}{1+4\sigma\pi_0} & \text{if}\ r=C\ \text{and}\ Y=1 \\
        \frac{(1-4\sigma)\pi_0}{1-4\sigma\pi_0} & \text{if}\ r=C\ \text{and}\ Y=0
    \end{cases} \label{eq:posterior}
\end{align}
about $\theta$ after observing the chosen rule $r$ and the outcome $Y$.

For convenience, I let
\[ \pi_1(Y)\equiv\Pr(\theta=1\,\vert\,r=C,Y) \]
denote the random posterior belief about $\theta$ when the expert chooses the complex rule.
This belief has expected value
\begin{align}
    \E[\pi_1(Y)]
    &= \pi_1(1)\Pr(Y=1\,\vert\,r=C)+\pi_1(0)\Pr(Y=0\,\vert\,r=C) \notag \\
    &= \Pr(Y=1\,\vert\,r=C,\theta=1)\pi_0+\Pr(Y=0\,\vert\,r=C,\theta=1)\pi_0 \notag \\
    &= \pi_0. \label{eq:posterior-mean}
\end{align}
Thus, choosing the complex rule leads to unchanged beliefs \emph{on average}, while choosing the simple rule leads to unchanged beliefs independently of whether the advisee takes the correct action.

\paragraph{Payoffs}

The expert's payoff has two components:
\begin{enumerate}

    \item
    A unit ``accuracy payoff'' from the advisee taking the correct action;
    
    \item
    A ``reputation payoff'' $\psi(\Pr(\theta=1\,\vert\,r,Y))$ that depends on the posterior belief~\eqref{eq:posterior} about $\theta$.

\end{enumerate}
I assume $\psi:[0,1]\to\R$ is non-decreasing but has $\psi(0)<\psi(1)$.
This function captures the expert's extrinsic incentive to be viewed as competent.
This incentive may reflect the continuation value of employment in a dynamic setting with replaceable experts (see Section~\ref{sec:wage-example}).
Alternatively, it may reflect public officials (e.g., Anthony Fauci) wanting to foster trust in government institutions (e.g., the Centres for Disease Control).

The expert's total payoff equals the sum of his accuracy and reputation payoffs.
He chooses the rule $r\in\{S,C\}$ that maximizes his expected payoff
\[ \Phi(r,\sigma,\pi_0)\equiv \begin{cases}
    0.5+\psi(\pi_0) & \text{if}\ r=S \\
    0.5+2\sigma\pi_0+\Psi(\sigma,\pi_0) & \text{if}\ r=C,
\end{cases} \]
where
\begin{align*}
    \Psi(\sigma,\pi_0)
    &\equiv \E[\psi(\pi_1(Y))] \\
    &= (0.5+2\sigma\pi_0)\psi\left(\frac{(1+4\sigma)\pi_0}{1+4\sigma\pi_0}\right)+(0.5-2\sigma\pi_0)\psi\left(\frac{(1-4\sigma)\pi_0}{1-4\sigma\pi_0}\right)
\end{align*}
is his expected reputation payoff from choosing the complex rule.
He breaks ties by choosing the complex rule.%
\footnote{
This is an arbitrary assumption.
It is not consequential for any of my results.
}

The advisee's payoff equals $Y$.
Therefore, she always follows the expert's advice.

\paragraph{Timing}
The model timing is as follows:
\begin{enumerate}

    \item
    The expert chooses the simple or complex rule.

    \item
    If the expert chooses the simple rule then the advisee takes the action.
    If he chooses the complex rule then she observes $X_\theta$ and takes the action if and only if $X_\theta=1$.

    \item
    The expert and advisee observe $Y$, and form a posterior belief~\eqref{eq:posterior} about $\theta$.

    \item
    Payoffs are realized.

\end{enumerate}

\subsection{Discussion of modeling assumptions}

\paragraph{Behavioral limitations}

I assume the advisee is boundedly rational.
Her behavioral limitation is not knowing the set of conditions that could predict whether the action is good.
This knowledge is what makes the expert an ``expert.''
However, if the expert provides a condition then I assume the advisee can verify its potential predictive power.
This prevents the expert from mis-reporting the predictive power of the condition he provides.
Thus, assuming verifiability allows me to focus on the expert's decision to provide the condition and abstract from his decision about how much uncertainty to communicate.

Intuitively, the expert in my model generates hypotheses about the relationship between conditions and actions.
Choosing the complex rule corresponds to stating a hypothesis and referencing the scientific literature on which it is based.
The advisee reads that literature and verifies that the hypothesis is ``reasonable.''
However, she cannot tell if the hypothesis is true.
That skill is unique to competent experts.

The form of bounded rationality I assume is intentionally unrealistic.
For example, one reason we outsource vaccine guidance to public health agencies, rather than have the public engage with the relevant virological science themselves, is because that science is difficult to understand.
This difficulty may lead the public to endure cognitive costs or make mistakes that agencies want to prevent.
Thus, one explanation for why experts provide ``simple'' advice is that they internalize the advisee's cognitive burden and avoid sharing ``too much detail.''
However, this explanation is unsatisfying because one can rationalize any amount of detail suppression with a suitable assumption about the advisee's cognitive costs.
In contrast, my reputation-based explanation does not appeal to cognitive costs.
It is meant to be idealized: to show that experts may simplify their advice even if the advisee has \emph{zero} cognitive costs.

Section~\ref{sec:measurement-errors} considers an extension to my model in which the advisee makes mistakes when implementing the complex rule.

\paragraph{Expert does not know $\theta$}

I assume the expert does not know if he is competent.
This allows me to abstract from competence signaling motives because the expert has no information about his competence to signal.
It makes the expert's choice of model uninformative about his competence, and isolates the tension between predictive power and reputational risk.
Section~\ref{sec:known-type} considers the implications of letting the expert know $\theta$.

\paragraph{Expert has aligned preferences}

I assume the expert is ``unbiased'' in that he wants the advisee to take the action if and only if she thinks it is good.
This assumption is appropriate in the science communication setting that motivates my analysis: experts are intrinsically ``benevolent'' but face extrinsic incentives to preserve their reputation.
It allows me to abstract from persuasion motives \cite[as studied, e.g., by][]{Kamenica-Gentzkow-2011-AER} and lie detection \cite[see, e.g.,][]{Balbuzanov-2019-IJGT}.


\section{Analysis}
\label{sec:model-analysis}

\subsection{Expert's choice}

Proposition~\ref{prop:rule-choice} characterizes the expert's choice between the simple and complex rules.
For convenience, I let
\begin{align*}
    \Delta\Phi(\sigma,\pi_0)
    &\equiv \Phi(C,\sigma,\pi_0)-\Phi(S,\sigma,\pi_0) \notag \\
    &= 2\sigma\pi_0+\Psi(\sigma,\pi_0)-\psi(\pi_0)
\end{align*}
denote the gain in the expert's expected payoff from choosing the complex rule over the simple rule.
He chooses the complex rule if and only if this gain is non-negative:

\begin{proposition}
    \label{prop:rule-choice}
    The expert chooses the complex rule if and only if
    \begin{equation}
        2\sigma\pi_0 \ge \psi(\pi_0)-\Psi(\sigma,\pi_0). \label{ineq:rule-choice}
    \end{equation}
\end{proposition}

The left-hand side of~\eqref{ineq:rule-choice} equals the gain in the expert's expected accuracy payoff from choosing the complex rule over the simple rule.
The right-hand side equals the loss in his expected reputation payoff from choosing the complex rule over the simple rule.

If $\psi$ is convex then the expert's reputation gain from advising the correct action exceeds his reputation loss from advising the incorrect action.
As a result, he chooses the complex rule because it is more accurate and offers an expected reputation \emph{benefit}:

\begin{corollary}
    \label{crly:rule-choice-convex}
    If $\psi$ is linear or convex, then the expert chooses the complex rule.
\end{corollary}

In contrast, if $\psi$ is concave then the reputation gain from advising the correct action is smaller than the reputation loss from advising the incorrect action.
As a result, choosing the complex rule incurs an expected reputation \emph{cost}.
Choosing the simple rule avoids this cost because it prevents the advisee from learning about the expert's competence.
He chooses the simple rule if this incentive to hide his competence dominates his incentive to advise the correct action.

Suppose the expert's reputation payoff comprises a wage $w\ge0$ earned with some probability $R(\pi)$ that depends on the posterior belief $\pi$ about his competence.
For example, the advisee may decide whether to re-hire the expert based on his perceived competence relative to other experts (see Section~\ref{sec:wage-example}).
Corollary~\ref{crly:rule-choice-wR} states that if $R$ is concave in $\pi$ (e.g., since competition among more competent experts is more aggressive) then the expert chooses the complex rule if and only if $w$ is sufficiently low.

\begin{corollary}
    \label{crly:rule-choice-wR}
    Let $w\ge0$ and suppose $\psi(\pi)=w\,R(\pi)$ for all beliefs $\pi\in[0,1]$, where $R:[0,1]\to\R$ is non-decreasing and concave.
    Then the expert chooses the complex rule if and only if
    \[ w\le\frac{2\sigma\pi_0}{R(\pi_0)-\E[R(\pi_1(Y))]}. \]
\end{corollary}

\begin{figure}
    \centering
    \includegraphics{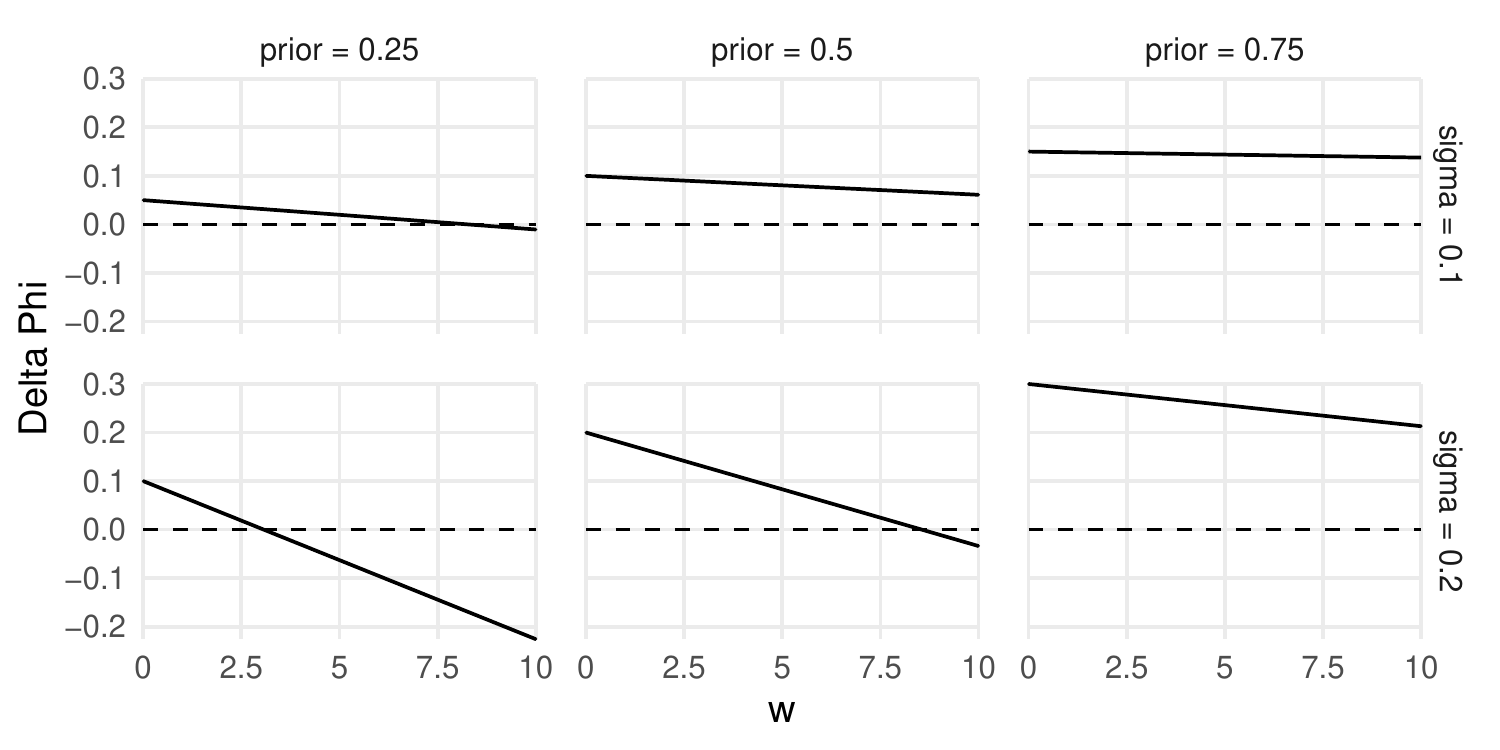}
    \caption{Values of $\Delta\Phi(\sigma,\pi_0)$ when $\psi(\pi)=w\sqrt\pi$, $\sigma\in\{0.1,0.2\}$, $\pi_0\in\{0.25,0.5,0.75\}$, and $w\in[0,10]$}
    \label{fig:sqrt-example-delta-Phi}
\end{figure}
For example, suppose $R(\pi)\equiv\sqrt{\pi}$.
Figure~\ref{fig:sqrt-example-delta-Phi} plots the corresponding value of~$\Delta\Phi(\sigma,\pi_0)$ for combinations of $\sigma\in\{0.1,0.2\}$, $\pi_0\in\{0.25,0.5,0.75\}$, and $w\in[0,10]$.
This value falls when $w$ rises because the expected reputational cost of choosing the complex rule rises.
If $\sigma=0.2$ and $\pi_0=0.25$ (i.e., competence is valuable but rare) then the complex rule is good at diagnosing the expert's probable incompetence.
Consequently, the expert chooses the complex rule if and only if $w<3.07$; that is, his reputational concerns are sufficiently \emph{weak}.
This threshold value of $w$ rises when $\sigma$ falls to 0.1, which makes competence less valuable but also makes the complex rule less diagnostic.
The threshold also rises when $\pi_0$ rises, in part because $\psi$ becomes less concave near $\pi_0$.
This fall in local concavity lowers the reputational loss from advising the incorrect action relative to the gain from advising the correct action.

In this example, the expert chooses the complex rule only if the maximum reputation payoff~$w$ is many times larger than his accuracy payoff.
This difference in payoff magnitudes is plausible if $w$ represents the continuation value of employment.
This value includes the present value of the accuracy payoffs accrued from future advisory events.
If there are many such events and the expert is patient then $w$ can be large.
For example, if he offers advice at each date $t\in\{0,1,2,\ldots,10\}$ and has inter-temporal discount factor $\delta=0.95$ then the date~0 value of his date $t\ge1$ accuracy payoffs is $\delta(1-\delta^{10})/(1-\delta)\approx7.62$ times larger than the value of his date~0 accuracy payoff.

\subsection{Changes in parameter values}

I now analyze how the posterior belief $\pi_1(Y)$ and gain in expected payoff $\Delta\Phi(\sigma,\pi_0)$ from choosing the complex rule depend on the parameters $\sigma$ and $\pi_0$.

\begin{lemma}
    \label{lem:posterior-comparative-statics}
    $\pi_1(1)$ is increasing in $\sigma$ and $\pi_0$, whereas $\pi_1(0)$ is decreasing in $\sigma$ and non-decreasing in $\pi_0$.
\end{lemma}

\begin{figure}
    \centering
    \includegraphics{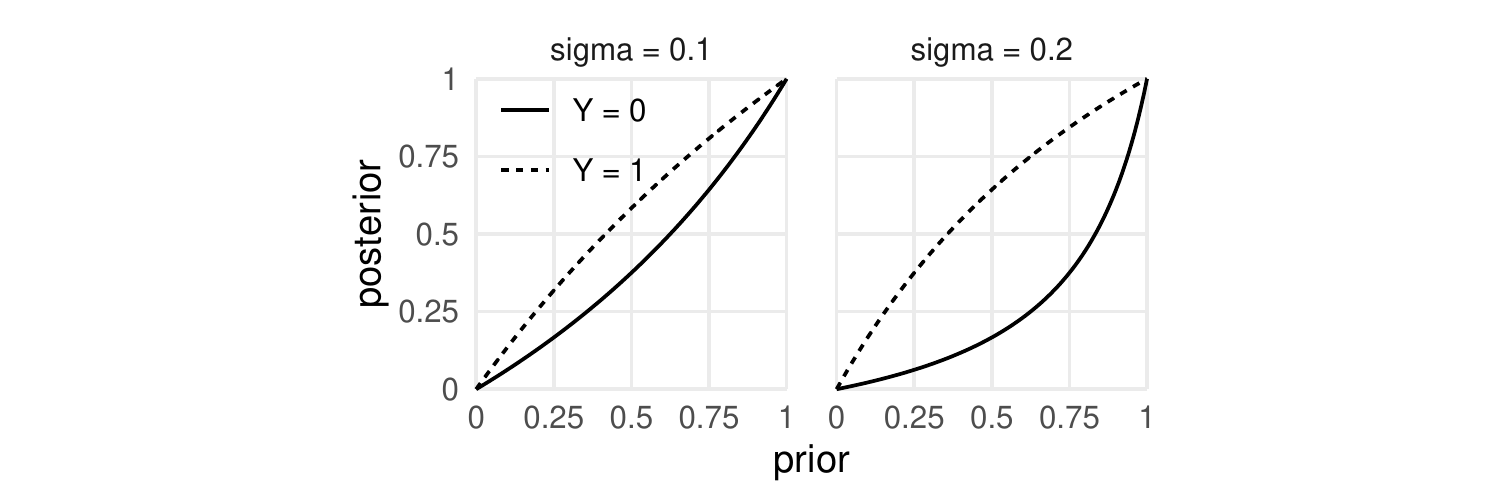}
    \caption{Posterior belief $\pi_1(Y)$ for $Y\in\{0,1\}$ and $\sigma\in\{0.1,0.2\}$}
    \label{fig:posteriors}
\end{figure}
Figure~\ref{fig:posteriors} shows how $\pi_1(1)$ and $\pi_1(0)$ grow with $\pi_0$ when $\sigma\in\{0.1,0.2\}$.
As $\sigma$ rises, the complex rule becomes less likely to lead to the incorrect action if the expert is competent.
Then observing $Y=0$ becomes more diagnostic of his \emph{in}competence if he chooses the complex rule.
Whether this dissuades him from choosing that rule depends on the shape of $\psi$:

\begin{proposition}
    \label{prop:Delta-Phi-comparative-statics-sigma}
    If $\psi$ is linear or convex, then $\Delta\Phi(\sigma,\pi_0)$ is increasing in $\sigma$.
\end{proposition}

Increasing the covariance $\sigma$ makes the complex rule (i)~more likely to lead to the correct action and (ii)~more diagnostic of the expert's competence.
These two forces on $\Delta\Phi(\sigma,\pi_0)$ act in the same direction when $\psi$ is convex.
In that case, the expert gains more from being seen as more competent than he loses from being seen as less competent.
Consequently, making the complex rule better at diagnosing his competence makes him more willing to choose that rule.

In contrast, forces (i) and (ii) act in \emph{opposite} directions when $\psi$ is concave.
In that case, the expert gains \emph{less} from being seen as more competent than he loses from being seen as less competent.
Consequently, making the complex rule better at diagnosing his competence makes him more willing to choose that rule if and only if his reputation concerns are sufficiently weak.
Figure~\ref{fig:sqrt-example-delta-Phi} demonstrates this condition: if $\psi(\pi)=w\sqrt{\pi}$ then $\Delta\Phi(0.2,0.5)>\Delta\Phi(0.1,0.5)$ when $w$ is small but $\Delta\Phi(0.2,0.5)<\Delta\Phi(0.1,0.5)$ when $w$ is large.

Increasing $\pi_0$ raises the probability that the complex rule has superior predictive power.
It also raises the expert's reputation payoff at each posterior belief about his competence.
However, this rise in reputation payoffs applies to the complex rule \emph{and} the simple model.
Consequently, the overall effect on the expert's willingness to choose the complex rule depends on how changing $\pi_0$ changes his reputation payoffs under each rule.
Proposition~\ref{prop:Delta-Phi-comparative-statics-prior-differentiable} characterizes this overall effect.

\begin{proposition}
    \label{prop:Delta-Phi-comparative-statics-prior-differentiable}
    Suppose $\psi$ is differentiable at $\pi_0$, $\pi_1(1)$, and $\pi_1(0)$.
    Then $\Delta\Phi(\sigma,\pi_0)$ is increasing in $\pi_0$ if and only if
    \begin{equation}
        2\sigma+2\sigma\left[\psi(\pi_1(1))-\psi(\pi_1(0))\right]+\E\left[\psi'(\pi_1(Y))\parfrac{\pi_1(Y)}{\pi_0}\right]>\psi'(\pi_0), \label{ineq:Delta-Phi-comparative-statics-prior-differentiable}
    \end{equation}
    where $\psi'$ denotes the derivative of $\psi$.
\end{proposition}

The left-hand side of~\eqref{ineq:Delta-Phi-comparative-statics-prior-differentiable} decomposes the effect of increasing $\pi_0$ on the expert's expected payoff from choosing the complex rule.
The first term equals the increase in his expected accuracy payoff.
The second term equals the increase in his expected reputational payoff from being more likely to be competent.
The third term equals the increase in his expected reputation payoff induced by the increase in the posteriors $\pi_1(1)$ and $\pi_1(0)$.
The right-hand side of~\eqref{ineq:Delta-Phi-comparative-statics-prior-differentiable} equals the increase in the reputation payoff from choosing the simple rule.
Increasing $\pi_0$ increases $\Delta\Phi(\sigma,\pi_0)$ if and only if the sum of the three terms on the left-hand side of~\eqref{ineq:Delta-Phi-comparative-statics-prior-differentiable} exceeds the term on the right-hand side.

Proposition~\ref{prop:Delta-Phi-comparative-statics-prior-wR} analyzes the class of reputation payoff functions $\psi(\pi)=w\,R(\pi)$ considered in Corollary~\ref{crly:rule-choice-wR}, replacing the concavity restriction with a differentiability restriction.

\begin{proposition}
    \label{prop:Delta-Phi-comparative-statics-prior-wR}
    Let $w\ge0$ and suppose $\psi(\pi)=w\,R(\pi)$ for all beliefs $\pi\in[0,1]$, where $R:[0,1]\to\R$ is non-decreasing and differentiable.
    Then either
    \begin{enumerate}

        \item
        $\Delta\Phi(\sigma,\pi_0)$ is increasing in $\pi_0$, or

        \item
        There exists $w^*>0$ such that $\Delta\Phi(\sigma,\pi_0)$ is increasing in $\pi_0$ if and only if $w<w^*$.

    \end{enumerate}
\end{proposition}

Intuitively, increasing $\pi_0$ raises the payoff from choosing the simple rule by a \emph{certain} amount but the payoff from choosing the complex rule by an \emph{uncertain} amount.
If $R$ is ``not too concave'' or the stakes are sufficiently low (i.e., $w$ is small) then the uncertain amount is more attractive.


\section{Wages and replacement thresholds}
\label{sec:wage-example}

This section considers a setting in which the expert earns wage $w>0$ if and only if the posterior belief about his competence exceeds a threshold $\pi^*\in(0,1)$.
This setting arises, for example, when the advisee faces an action in each of two periods, and she hires the expert to provide advice in the second period if and only if she believes he is more likely to be competent than another expert drawn from the market.
Then the expert's reputation payoff is given by
\begin{equation}
    \psi(\pi)=\begin{cases}
        w & \text{if}\ \pi\ge\pi^* \\
        0 & \text{otherwise}
    \end{cases} \label{eq:wage-example-psi}
\end{equation}
for all beliefs $\pi\in[0,1]$.
Lemma~\ref{lem:wage-example-expected-payoff} defines the expert's expected payoffs from choosing the simple and complex rules in this setting.

\begin{lemma}
    \label{lem:wage-example-expected-payoff}
    Suppose the reputation payoff function $\psi$ is defined as in~\eqref{eq:wage-example-psi}.
    Then the expert's expected payoff from choosing the complex rule equals
    \begin{equation}
            \Phi(C,\sigma,\pi_0)=\begin{cases}
            0.5+2\sigma\pi_0 & \text{if}\ \pi_0<\underline\pi \\
            0.5+2\sigma\pi_0+w(0.5+2\sigma\pi_0) & \text{if}\ \underline\pi\le\pi_0<\overline\pi \\
            0.5+2\sigma\pi_0+w & \text{if}\ \overline\pi\le\pi_0,
        \end{cases} \label{eq:wage-example-expected-payoff-complex}
    \end{equation}
    where
    \begin{equation}
        \underline\pi\equiv\frac{\pi^*}{1+4\sigma(1-\pi^*)} \label{eq:wage-example-lower-threshold}
    \end{equation}
    and
    \begin{equation}
        \overline\pi\equiv\frac{\pi^*}{1-4\sigma(1-\pi^*)}. \label{eq:wage-example-upper-threshold}
    \end{equation}
    His expected payoff from choosing the simple rule equals
    \begin{equation}
        \Phi(S,\sigma,\pi_0)=\begin{cases}
            0.5+w & \text{if}\ \pi_0\ge\pi^* \\
            0.5 & \text{otherwise}.
        \end{cases} \label{eq:wage-example-expected-payoff-simple}
    \end{equation}
\end{lemma}

Suppose the expert chooses the simple rule.
Then the posterior belief about his competence equals the prior belief because no information about his competence is revealed.
Consequently, he receives the wage $w$ if and only if his prior exceeds the threshold $\pi^*$.

Now suppose the expert chooses the complex rule.
If $\pi_0<\underline\pi$ then the expert is guaranteed a reputation payoff of zero because the posterior belief $\pi_1(1)$ induced by the advisee taking the correct action is less than $\pi^*$.
Conversely, if $\pi_0\ge\underline\pi$ then the expert is guaranteed a reputation payoff of $w$ because the posterior belief $\pi_1(0)$ induced by the advisee taking the \emph{incorrect} action is greater than $\pi^*$.
If $\pi_0\in[\underline\pi,\overline\pi)$ then the expert receives a positive reputation payoff if and only if the advisee takes the correct action.

\begin{figure}
    \centering
    \includegraphics{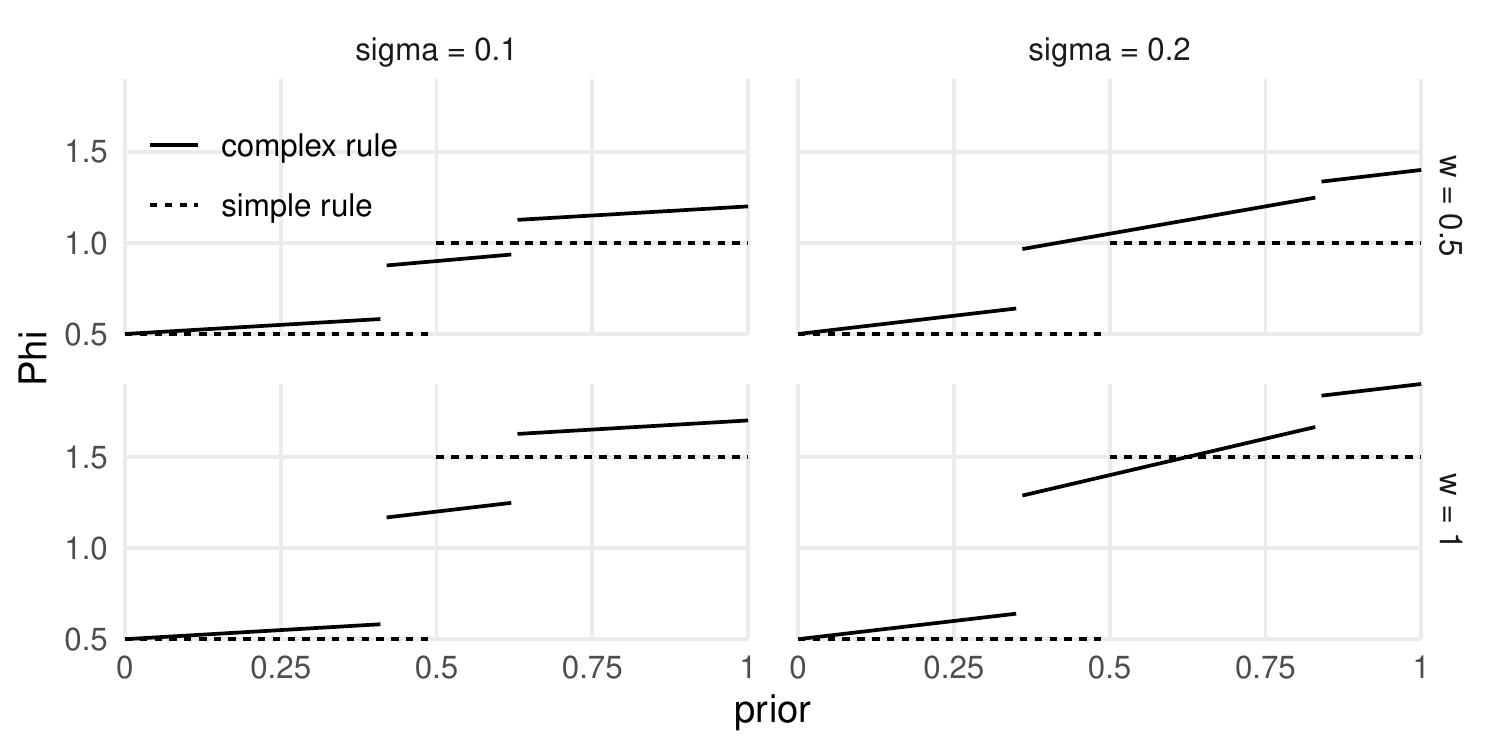}
    \caption{Graphs of $\Phi(C,\sigma,\pi_0)$ and $\Phi(S,\sigma,\pi_0)$ with $\psi$ defined as in~\eqref{eq:wage-example-psi} and with $\sigma\in\{0,1\}$, $w\in\{0.5,1\}$, and $\pi^*=0.5$}
    \label{fig:wage-example}
\end{figure}
Figure~\ref{fig:wage-example} plots the expected payoffs~\eqref{eq:wage-example-expected-payoff-complex} and~\eqref{eq:wage-example-expected-payoff-simple} against $\pi_0$ when $\sigma\in\{0.1,0.2\}$, $w\in\{0.5,1\}$, and $\pi^*=0.5$.
It demonstrates three cases:
\begin{enumerate}

    \item
    $\Phi(C,\sigma,\pi_0)$ exceeds $\Phi(S,\sigma,\pi_0)$ for all $\pi_0$ (as when, e.g., $(\sigma,w)=(0.2,0.5)$).

    \item
    $\Phi(C,\sigma,\pi_0)$ and $\Phi(S,\sigma,\pi_0)$ ``cross'' at the replacement threshold $\pi_0=\pi^*$ and the discontinuity $\pi_0=\overline\pi$ (as when, e.g., $(\sigma,w)=(0.1,0.5)$);

    \item
    $\Phi(C,\sigma,\pi_0)$ and $\Phi(S,\sigma,\pi_0)$ intersect at some $\pi_0\in[\underline\pi,\overline\pi)$ (as when, e.g., $(\sigma,w)=(0.2,1)$).

\end{enumerate}
Proposition~\ref{prop:wage-example-choice} states that these are the \emph{only} three possible cases.%
\footnote{
Corollary~\ref{crly:rule-choice-wR} does not apply in this setting because $\psi$ is not concave.
}

\begin{proposition}
    \label{prop:wage-example-choice}
    Suppose the reputation payoff function $\psi$ is defined as in~\eqref{eq:wage-example-psi}.
    Define
    \begin{equation}
        \pi^\dagger\equiv\frac{w}{4\sigma(1+w)}. \label{eq:wage-example-prior-dagger}
    \end{equation}
    If $\pi^\dagger\le\pi^*$ then the expert chooses the complex rule.
    If $\pi^\dagger>\pi^*$ then he chooses the \emph{simple} rule if and only if $\pi^*\le\pi_0<\min\{\pi^\dagger,\overline\pi\}$, where $\overline\pi$ is defined as in Lemma~\ref{lem:wage-example-expected-payoff}.
\end{proposition}

If the wage $w$ is sufficiently low or the threshold $\pi^*$ is sufficiently high, then the expert receives a greater expected payoff from choosing the complex rule regardless of the prior $\pi_0$.
If $w$ is high and the prior lies between $\pi^*$ and $\min\{\pi^\dagger,\overline\pi\}$, then the reputational risk induced by choosing the complex rule dissuades the expert from making that choice.

If $\pi_0$ is the proportion of employable experts who are competent and the advisee replaces the expert with a random competitor, then she maximizes the probability of employing a competent expert by setting $\pi^*=\pi_0$.
In that case, if she wants the expert to choose the complex model then she must limit his reputational risk by offering a low wage:

\begin{corollary}
    \label{crly:wage-example-wage-threshold}
    Suppose the reputation payoff function $\psi$ is defined as in~\eqref{eq:wage-example-psi}.
    If $\pi^*=\pi_0$ then the expert chooses the complex rule if and only if $w\le4\sigma\pi_0/(1-4\sigma\pi_0)$.
\end{corollary}

Alternatively, the wage may be fixed (due to, e.g., legal or institutional constraints).
Then the advisee's only way to solicit the complex rule is to vary the replacement threshold $\pi^*$.
She has two choices: set it low enough to guarantee that the expert receives the wage, or set it \emph{high} enough to guarantee that expert \emph{does not} receive the wage.
Both choices remove the expert's reputational risk.
The first choice removes the risk by fully insuring it.
The second choice removes the risk by making explicit that the advisee seeks the expert's advice as a ``one-off,'' effectively imposing a single-period term limit on the expert's employment.


\section{Extensions}
\label{sec:extensions}

This section considers several standalone extensions to my analysis in Sections~\ref{sec:model-setup} and~\ref{sec:model-analysis}.

\subsection{$\Pr(A=1)\not=0.5$}
\label{sec:prob-A}

Sections~\ref{sec:model-setup} and~\ref{sec:model-analysis} assume $\Pr(A=1)=0.5$.
This equates the simple and complex rules' expected accuracy payoffs when the expert is incompetent.
Consequently, it isolates the tension between the complex rule's superior predictive power and its ability to reveal the expert's competence.
If $\Pr(A=1)=a\in(0,1)$ then
\[ \Pr(A=X_\theta)=ax+(1-a)(1-x)+2\theta\sigma \]
and the condition~\eqref{ineq:rule-choice} under which the expert chooses the complex rule becomes
\begin{equation}
    2\sigma\pi_0-(2a-1)(1-x)\ge\psi(\pi_0)-\Psi(\sigma,\pi_0). \label{eq:rule-choice-a}
\end{equation}
The left-hand side of~\eqref{ineq:rule-choice} exceeds the left-hand side of~\eqref{eq:rule-choice-a} if and only if $a>0.5$.
In that case, conditioning on $X_0$ lowers the expert's accuracy payoff because it lowers the probability that the advisee takes the action when it is good.
Conditioning on $X_1$ raises the expert's accuracy payoff if and only if the covariance $\sigma$ of $A$ and $X_1$ is sufficiently high.
This covariance depends on $a$ since it depends on the variance $\Var(A)=a(1-a)$.
Therefore, increasing $a$ changes both terms on the left-hand side of~\eqref{eq:rule-choice-a}.%
\footnote{
If $A$ and $X_1$ have correlation $\rho$ then $\sigma=\rho\sqrt{a(1-a)x(1-x)}$.
Then the left-hand side of~\eqref{eq:rule-choice-a} is increasing in $a$ if and only if
\[ \rho\pi_0(1-2a)>2(1-x)\sqrt{\frac{a(1-a)}{x(1-x)}}. \]
If $a\ge0.5$ then this condition never holds.
If $a<0.5$ then it holds if and only if $\rho\pi_0$ is sufficiently large.
}
Increasing $a$ also changes the right-hand side by changing the distribution of posterior beliefs $\pi_1(Y)$ and, in turn, the expected reputation payoff $\Psi(\sigma,\pi_0)$.
The overall effect of increasing $a$ on $\Delta\Phi(\sigma,\pi_0)$ is ambiguous; it depends on the parameters $x$, $\sigma$ and $\pi_0$, and the shape of $\psi$.
I do not elaborate on this dependency because it is not central to my analysis.

\subsection{Measurement errors}
\label{sec:measurement-errors}

One reason real-world experts may prefer the simple rule is that the advisee may make mistakes when implementing the complex rule.
For example, suppose the expert provides condition $X_\theta$ but the advisee observes $\hat{X}_\theta\in\{0,1\}$ with $\Pr(\hat{X}_\theta\not=X_\theta)=\epsilon\in(0,0.5]$.
Proposition~\ref{prop:measurement-errors-complex-accuracy} states that such measurement errors decrease the complex rule's accuracy.%
\footnote{
If $\Pr(A=1)\not=0.5$ then whether measurement errors decrease the complex rule's accuracy depends on the joint distribution of $A$ and $X_\theta$, and on the relative rates of false positives and negatives.
}

\begin{proposition}
    \label{prop:measurement-errors-complex-accuracy}
    Let $\epsilon\in[0,0.5]$.
    Suppose $\hat{X}_\theta\in\{0,1\}$ is distributed such that $\Pr(\hat{X}_\theta\not=X_\theta)=\epsilon$ for each $\theta\in\{0,1\}$.
    Then 
    \begin{equation}
        \Pr(A=\hat{X}_\theta)=0.5+2\theta(1-2\epsilon)\sigma. \label{eq:measurement-errors-complex-accuracy}
    \end{equation}
\end{proposition}

Measurement errors decrease the covariance between the observed condition and the correct action.
By Proposition~\ref{prop:Delta-Phi-comparative-statics-sigma}, this decrease makes the expert less willing to choose the complex rule if the reputation payoff function $\psi$ is linear or convex.
However, if $\psi$ is concave then measurement errors increase the expected reputation payoff from choosing the complex rule because they make that rule less diagnostic of the expert's competence.
Thus, introducing measurement errors may lead the expert to switch to the complex rule if the reputation benefit exceeds the accuracy cost.

In practice, measurement errors may arise exogenously from, for example, imprecise medical tests.
They may also arise endogenously from the advisee optimally allocating their scarce attention and mental resources.%
\footnote{
This possibility relates to \citeapos{Oprea-2020-AER} claim that complex rules are cognitively costly to implement.
}
For example, suppose the advisee receives benefit $b>0$ from taking the correct action and zero from taking the incorrect action.
Suppose further that she endures cost $c(0.5-\epsilon)^2$ from observing $\hat{X}_\theta$ with error $\epsilon\in[0,0.5]$, where $c\in(0,2b\sigma\pi_0)$.%
\footnote{
This upper bound on $c$ ensures $\epsilon^*<0.5$.
}
If the expert chooses the complex rule then the advisee chooses the error $\epsilon^*$ that maximizes
\[ b\left(0.5+2\sigma(1-2\epsilon)\pi_0\right)-c(0.5-\epsilon)^2, \]
which has unique maximizer
\[ \epsilon^*\equiv0.5-\frac{2b\sigma\pi_0}{c}. \]
Substituting $\epsilon=\epsilon^*$ into~\eqref{eq:measurement-errors-complex-accuracy} then gives
\[ \Pr(A=\hat{X}_\theta)=0.5+\frac{4b\theta\sigma^2\pi_0}{c}. \]
Thus, in this example, the accuracy of the complex rule is increasing in the advisee's benefit-cost ratio $b/c$.
Now suppose the expert knows $b$ and $c$, and anticipates the advisee's choice of error $\epsilon^*$ when choosing between rules.
Then it is possible the advisee is willing to ``pay'' for low errors because her benefit-cost ratio is high, but this payment leads the expert to avoid the complex rule because it becomes too diagnostic of his competence.
If the advisee wants to learn the expert's competence (e.g., because she will face a higher stakes choice in the future) then she may prefer to under-invest in measuring the condition correctly.
This would encourage the expert to choose the complex rule and reveal information about his competence.

\subsection{Expert knows $\theta$}
\label{sec:known-type}

Sections~\ref{sec:model-setup} and~\ref{sec:model-analysis} assume that the expert does not know whether he is competent.
This assumption allows me to isolate his trade-off between providing accurate advice and managing reputational risk.
If the expert knows his competence type then he also has a signaling motive.
I analyze that motive below.

Suppose the expert knows his competence type $\theta$ but the advisee does not.
Then his choice of rule may reveal information about $\theta$.
Lemma~\ref{lem:known-type-posterior} describes how the advisee incorporates that information into her posterior belief about $\theta$.

\begin{lemma}
    \label{lem:known-type-posterior}
    Suppose the expert knows $\theta$ but the advisee does not.
    Let $\hat{q}_\theta\equiv\Pr(r=C\,\vert\,\theta)$ be the probability with which the advisee believes an expert with type $\theta$ chooses the complex rule.
    Then her posterior belief about $\theta$ after observing the chosen rule $r\in\{S,C\}$ and its outcome $Y\in\{0,1\}$ is given by
    \begin{equation}
        \Pr(\theta=1\,\vert\,r,Y)
        = \begin{cases}
            \pi_S & \text{if}\ r=S \\
            \frac{(1+4\sigma)\pi_C}{1+4\sigma\pi_C} & \text{if}\ r=C\ \text{and}\ Y=1 \\
            \frac{(1-4\sigma)\pi_C}{1-4\sigma\pi_C} & \text{if}\ r=C\ \text{and}\ Y=0,
        \end{cases} \label{eq:known-type-posterior}
    \end{equation}
    where
    \begin{equation}
        \pi_r=\begin{cases}
            \frac{(1-\hat{q}_1)\pi_0}{(1-\hat{q}_1)\pi_0+(1-\hat{q}_0)(1-\pi_0)} & \text{if}\ r=S \\
            \frac{\hat{q}_1\pi_0}{\hat{q}_1\pi_0+\hat{q}_0(1-\pi_0)} & \text{if}\ r=C
        \end{cases} \label{eq:known-type-intermediate-belief}
    \end{equation}
    is her ``intermediate'' belief about $\theta$ after observing $r$ but before observing $Y$.
\end{lemma}

Equations~\eqref{eq:posterior} and~\eqref{eq:known-type-posterior} coincide when $p_0=p_1$; that is, when the advisee believes that competent and incompetent experts choose between rules with the same probabilities.
Then the choice of rule is uninformative and so $\pi_r=\pi_0$ for each $r\in\{S,C\}$.

I now study the perfect Bayesian equilibrium (henceforth ``equilibrium'') of the game in which the expert chooses a rule and the advisee forms a posterior belief about $\theta$.
In this game, the expert's strategy comprises (i)~a probability $p_\theta$ with which he chooses the complex rule given his type $\theta$ and (ii)~a conjecture $\hat{p}_{1-\theta}$ about the corresponding probability for the other type.
The advisee's strategy comprises a conjecture $\hat{q}_\theta$ about the probability for each type.
She uses this conjecture to form her posterior belief~\eqref{eq:known-type-posterior}.
In equilibrium, all conjectures are correct: $p_\theta=\hat{p}_\theta=\hat{q}_\theta$.
As a shorthand, I describe an equilibrium by the pair $(p_0,p_1)$ of probabilities that define it, and refer to $p_\theta$ as the $\theta$-type expert's ``strategy.''
This strategy is ``pure'' if $p_\theta\in\{0,1\}$ and ``mixed'' if $p_\theta\in(0,1)$.
It is his best-response to the strategy played by the other type, given the advisee's subsequent inference about $\theta$ and its reputational consequences.

Lemma~\ref{lem:known-type-expected-payoffs} defines the expected payoff that experts of each type maximize to determine their best-response.

\begin{lemma}
    \label{lem:known-type-expected-payoffs}
    Suppose the expert knows $\theta$ but the advisee does not.
    Let $p_\theta$ be the equilibrium probability with which an expert with type $\theta$ chooses the complex rule.
    Then such an expert has expected payoff
    \begin{equation}
        \Omega_\theta(\sigma,\pi_0,p_0,p_1)
        \equiv p_\theta\Phi_\theta(C,\sigma,\pi_0,p_0,p_1)+(1-p_\theta)\Phi_\theta(S,\sigma,\pi_0,p_0,p_1), \label{eq:known-type-expected-payoff}
    \end{equation}
    where
    \begin{equation}
        \Phi_\theta(C,\sigma,\pi_0,p_0,p_1)
        \equiv 0.5+2\sigma\theta+(0.5+2\sigma\theta)\psi\left(\frac{(1+4\sigma)\pi_C}{1+4\sigma\pi_C}\right)+(0.5-2\sigma\theta)\psi\left(\frac{(1-4\sigma)\pi_C}{1-4\sigma\pi_C}\right) \label{eq:known-type-expected-payoff-complex}
    \end{equation}
    is his expected payoff from choosing the complex rule with certainty,
    \begin{equation}
        \Phi_\theta(S,\sigma,\pi_0,p_0,p_1)\equiv 0.5+\psi\left(\pi_S\right) \label{eq:known-type-expected-payoff-simple}
    \end{equation}
    is his expected payoff from choosing the simple rule with certainty, and $\pi_C$ and $\pi_S$ are defined as in Lemma~\ref{lem:known-type-posterior}.
\end{lemma}

I first show that there are no pure separating equilibria.
This is because incompetent experts never play strategies that allow the advisee to diagnose their incompetence with certainty.

\begin{proposition}
    \label{prop:known-type-pure-equilibrium}
    Suppose the expert knows $\theta$ but the advisee does not.
    If a pure strategy equilibrium exists, then it is pooling.
\end{proposition}

The game may have two, one, or no pure pooling equilibria.
It may also have mixed strategy equilibria.
The number and nature of equilibria depend on the parameters $\sigma$ and $\pi_0$, and the shape of the reputation payoff function $\psi$.

To gain traction, I focus on the ``replaceable expert'' setting discussed in Section~\ref{sec:wage-example}.
I assume the replacement threshold equals $\pi^*=\pi_0$.
Lemma~\ref{lem:known-type-wage-example-expected-payoffs} defines the expert's expected payoffs in this setting.

\begin{lemma}
    \label{lem:known-type-wage-example-expected-payoffs}
    Suppose the expert knows $\theta$ but the advisee does not.
    Let the probabilities $p_\theta$ and functions $\Phi_\theta(r,\sigma,\pi_0,p_0,p_1)$ be defined as in Lemma~\ref{lem:known-type-expected-payoffs}, and let the reputation payoff function $\psi$ be defined as in~\eqref{eq:wage-example-psi} with $\pi^*=\pi_0$.
    Then
    \begin{equation}
        \Phi_\theta(C,\sigma,\pi_0,p_0,p_1)
        =\begin{cases}
            0.5+2\sigma\theta & \text{if}\ p_1(1+4\sigma)<p_0 \\
            0.5+2\sigma\theta+(0.5+2\sigma\theta)w & \text{if}\ p_1(1-4\sigma)<p_0\le p_1(1+4\sigma) \\
            0.5+2\sigma\theta+w & \text{if}\ p_0\le p_1(1-4\sigma)
        \end{cases} \label{eq:known-type-wage-example-expected-payoff-complex}
    \end{equation}
    and
    \begin{equation}
        \Phi_\theta(S,\sigma,\pi_0,p_0,p_1)
        =\begin{cases}
            0.5+w & \text{if}\ p_0\ge p_1 \\
            0.5 & \text{otherwise}.
        \end{cases} \label{eq:known-type-wage-example-expected-payoff-simple}
    \end{equation}
\end{lemma}

Proposition~\ref{prop:known-type-wage-example-best-responses} determines competent and incompetent experts' best responses by maximizing their expected payoff $\Omega_\theta(\sigma,\pi_0,p_0,p_1)$ under the conditions specified in Lemma~\ref{lem:known-type-wage-example-expected-payoffs}.

\begin{proposition}
    \label{prop:known-type-wage-example-best-responses}
    Under the conditions specified in Lemma~\ref{lem:known-type-wage-example-expected-payoffs},  we have
    \begin{equation}
        \argmax_p\Omega_0(\sigma,\pi_0,p,p_1)=\{p_1\} \label{eq:known-type-wage-example-best-response-incompetent}
    \end{equation}
    and
    \begin{align}
        &\argmax_p\Omega_1(\sigma,\pi_0,p_0,p) \notag \\
        &= \begin{cases}
            \{0\} & \text{if}\ 1-p_0<4\sigma<1\ \text{and}\ \max\left\{2\sigma,4\sigma/(1-4\sigma)\right\}<w \\
            \{0\}\cup\left[\frac{p_0}{1+4\sigma},p_0\right]\cup\{1\} & \text{if}\ 1-p_0<4\sigma<1\ \text{and}\ 2\sigma<w=4\sigma/(1-4\sigma) \\
            \{1\} & \text{otherwise}.
        \end{cases} \label{eq:known-type-wage-example-best-response-competent}
    \end{align}

\end{proposition}

Incompetent experts best-respond by mimicking competent experts.
Thus, in any equilibrium, the advisee learns nothing from the choice of rule because competent and incompetent experts play the same strategies.
Yet she may learn from the \emph{outcome} of the complex rule---i.e., whether it leads to the correct action---if the expert chooses it.
Competent experts have a dominant strategy to choose the complex rule if (i)~$w\le 2\sigma$ or (ii)~$4\sigma<1$ and $w<4\sigma/(1-4\sigma)$---i.e., their reputation concerns are sufficiently weak.
Otherwise, they want to differentiate themselves from incompetent experts.

\begin{figure}
    \centering
    \includegraphics{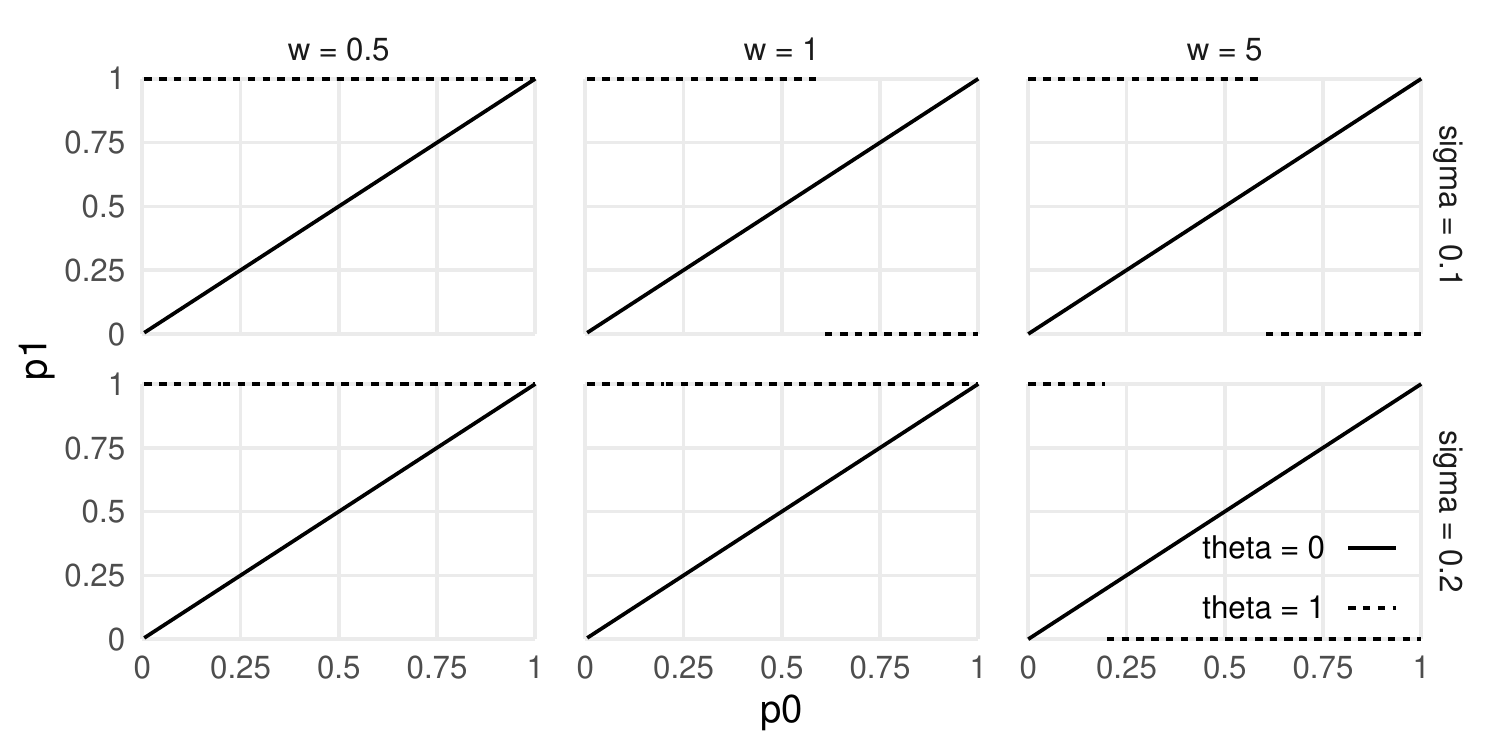}
    \caption{Experts' best-response curves when $\sigma\in\{0.1,0.2\}$ and $w\in\{0.5,1,5\}$}
    \label{fig:known-type-best-responses}
\end{figure}
Figure~\ref{fig:known-type-best-responses} shows the best response curves for experts of each type $\theta\in\{0,1\}$ when $\sigma\in\{0.1,0.2\}$ and $w\in\{0.5,1,5\}$.
These curves do not depend on the prior $\pi_0$.
Their intersections correspond to equilibria.
If the wage $w$ is small and the covariance $\sigma$ is large, then the unique equilibrium involves pure pooling on the complex rule.
Increasing $w$ or decreasing $\sigma$ destroys this equilibrium.
Indeed, Corollary~\ref{crly:known-type-wage-example-equilibria} states that if the wage $w$ is large relative to the covariance $\sigma$ then the game has no equilibria:

\begin{corollary}
    \label{crly:known-type-wage-example-equilibria}
    Under the conditions specified in Lemma~\ref{lem:known-type-wage-example-expected-payoffs},
    \begin{enumerate}

        \item[(i)]
        If $4\sigma<1$ and $\max\left\{2\sigma,4\sigma/(1-4\sigma)\right\}<w$, then there are no equilibria;

        \item[(ii)]
        If $4\sigma<1$ and $2\sigma<w=4\sigma/(1-4\sigma)$, then the pair $(p_0,p_1)$ defines an equilibrium if and only if $p_0=p_1\in(1-4\sigma,1]$.

    \end{enumerate}
    Otherwise, the pair $(p_0,p_1)=(1,1)$ defines the unique equilibrium.
\end{corollary}

Consider the knife-edge case in which the competent expert's condition is imperfectly predictive (i.e., $4\sigma<1$) and $2\sigma<w=4\sigma/(1-4\sigma)$.
If he can mimic an incompetent expert who chooses the complex rule with sufficiently high probability, then he is indifferent between doing so and choosing the complex rule for sure.
Otherwise, and outside this knife-edge case, the game either has no equilibrium or a unique equilibrium with pure pooling on the complex rule.
This pooling equilibrium exists only when the expert's reputational concerns (indexed by $w$) are small relative to his accuracy concerns (indexed by $\sigma$).
Thus, letting the expert know his type preserves the main insight from Sections~\ref{sec:model-setup}--\ref{sec:wage-example}: experts give complex advice only if they face low downside reputational risk.%
\footnote{
This insight echoes an insight from the innovation literature that failure tolerance tends to improve performance \citep[see, e.g.,][]{Azoulay-etal-2011-RAND,Manso-2011-JF,Tian-Wang-2014-RFStud}.
}


\section{Related literature}
\label{sec:literature}

My analysis is closely related to the literature on reputational cheap talk.
That literature explores how the incentive to appear informed distorts the equilibria of cheap talk games in which a sender shares his signal of an \emph{ex post} observable state.
\cite{Ottaviani-Sorensen-2006-JET,Ottaviani-Sorensen-2006-RAND} show that, in equilibrium, the sender does not report his signal truthfully if it depends on whether he is informed.
Instead, he manipulates his report to induce the most favorable belief about his informedness.
For example, he may bias his report toward the receiver's prior belief about the state to hedge against the report disagreeing with the state.%
\footnote{
\cite{Klein-Mylovanov-2017-JMathEcon} argue that this behavior is limited to models with short time horizons.
If the expert has a long horizon, then he may forgo ``impressing'' the receiver in the short run if he becomes privately pessimistic about his ability to impress in the long run.
}
This hedging behavior corresponds to choosing the simple rule in my model.
Indeed, one could interpret my model as an example of reputational cheap talk in which the expert reports information about the predictive power of the condition he identifies.
However, the expert in my model has reputational concerns associated with the correctness of the action taken as a consequence of his report, rather than the correctness of the report itself.%
\footnote{
This distinction between actions and information is the focus of \cite{Andrews-Shapiro-2021-ECTA}, who study statistical decision theoretic models of science communication.
}
The innovation of my paper is to show how reputational concerns influence the type of advice experts provide: do they make ``blanket recommendations,'' or do they provide conditions under which to take different actions?%
\footnote{
\cite{Levit-Tsoy-2020-AEJMicro} argue that experts make blanket recommendations to conceal their bias.
}

\cite{Backus-Little-2020-APSR} offer another model similar to mine: one of an unbiased expert who provides advice about an action but has reputational concerns.
In their model, the correct action may be impossible to identify and only competent experts know if identification is possible.
This makes competent experts averse to admitting ignorance because doing so makes them look incompetent.
Whereas in my model, the correct action \emph{is} possible to identify but only competent experts can identify it.
This makes competent experts \emph{pretend to be ignorant} (by choosing the simple rule) if the reputational cost of appearing incompetent is sufficiently large.
\citeauthor{Backus-Little-2020-APSR} and my conclusions differ because our definitions of ``competence'' differ.

My definition of ``expertise'' is similar to that used by \cite{Callander-etal-2021-JPE}.
They note that experts tend to provide ``referential'' advice that combines suggestions for what to do with examples of inferior alternatives.
The authors posit that expertise consists in knowing the relative merits of different suggestions.
Similarly, I posit that expertise consists in knowing conditions under which different suggestions are preferable.
\citeauthor{Callander-etal-2021-JPE} and my analyses differ in what motivates the expert.
They focus on the expert's desire to retain an informational advantage over their advisee.
In contrast, I allow the expert to forfeit this advantage because it makes his advice more likely to induce the correct the action and may yield reputational benefits.

\cite{Callander-etal-2021-JPE} build on the literature on disclosing verifiable information.
The strand of that literature closest to my analysis is due to \cite{Dye-1985-JAccRes}.
It considers the possibility that the sender is uninformed.
Then the receiver cannot distinguish between the sender withholding information and having none \citep[see, e.g.,][]{Dziuda-2011-JET}.
\citeauthor{Dye-1985-JAccRes} studies a model in which the receiver knows the sender has information but does not know its value.
This is similar to my model, in which the advisee knows the expert is aware of predictive conditions but does not know their predictive power.
\citeauthor{Dye-1985-JAccRes} and my models differ in that the expert I study has no notion of whether his information is ``good'' or ``bad'' for the advisee.
So it is not that competent experts ``can partially conceal bad news by pooling with those who are unable to disclose'' \cite[p.\! 226]{Grubb-2011-JEMS}.
Instead, competent and incompetent experts avoid disclosure because they want to avoid its reputational consequences.

My assumption that the expert does not know his competence type is typical in the literature on career concerns.
That literature explores how agents' desire to appear competent drives a wedge between their risk preferences and principals' \citep{Holmstrom-RicartiCosta-1986-QJE,Holmstrom-1999-REStud}.%
\footnote{
\cite{Dewatripont-etal-1999-REStud} generalize \citeapos{Holmstrom-1999-REStud} model through the lens of information design.
}
This wedge can lead to inefficient information seeking \citep{Levy-2004-EER,Milbourn-etal-2001-RAND,Scharfstein-Stein-1990-AER} and risk-taking \citep{Aghion-Jackson-2016-AEJMicro,Chevalier-Ellison-1997-JPE,Prendergast-Stole-1996-JPE}.
For example, \cite{Aghion-Jackson-2016-AEJMicro} analyze replacement schemes in a dynamic setting similar to that described in Section~\ref{sec:wage-example}.%
\footnote{
\cite{Tsuyuhaha-2012-EB} also analyzes this setting in a two-period model.
He focuses on replacing biased experts, whereas I focus on replacing incompetent experts.
}
They study a ``leader'' who takes a materially safe or risky action, where the cost of the risky action being bad exceeds the benefit of it being good.
This concavity in material payoffs drives the leader toward the safe action, which is less beneficial on average but reveals no information about his type.
In contrast, the expert in my model chooses between \emph{reputationally} safe and risky advice, and I show how his choice depends on the concavity in his reputational payoffs.%
\footnote{
\cite{Rappoport-2022-} also distinguishes between material and reputational payoffs.
}

Many insights from the career concerns literature apply to my model.
For example, \cite{Gibbons-Murphy-1992-JPE} explain that compensation contracts should be strongest for workers close to retirement because they have the weakest career concerns.
The advisee in my model may prefer an older expert with lower expected competence $\pi_0$ to a younger expert with higher expected competence.
This is because the older expert has a lower continuation value of employment, making him more willing to risk losing that value by choosing the complex rule.
I analyze the relationship between expected competence, continuation values, and rule choices in Propositions~\ref{prop:Delta-Phi-comparative-statics-prior-wR} and \ref{prop:wage-example-choice}.

More generally, the reputation payoff function $\psi$ acts as a reduced-form expression for the expert's compensation scheme.
Its shape determines the expert's choice.
For example, if $\psi$ is close to constant then the expert prefers to choose the complex rule and risk being ``wrong'' (see Corollary~\ref{crly:rule-choice-convex}).
This echoes another insight from the career concerns literature: weakening incentives, such as by offering tenure \citep{Aghion-Jackson-2016-AEJMicro} or severance pay \citep{Bonatti-Horner-2017-TE}, can lead to more efficient behavior.

\section{Conclusion}
\label{sec:conclusion}

This paper discusses how reputational concerns impact the type of advice that experts provide.
I show that experts may ``simply'' their advice by excluding relevant conditions if their payoff is sufficiently concave in posterior beliefs about their competence.
This simplification makes advice worse by pooling states in which actions should and should not be taken.
Advisees can prevent experts from giving simple advice by limiting their reputational risk: by reducing the payoff from being perceived as competent, or by making payoffs independence of perceived competence.


{
\small
\raggedright
\bibliographystyle{apalike}
\bibliography{main}
}

\clearpage

\appendix

\section{Proofs}
\label{sec:proofs}

\subsection{Results in Section~\ref{sec:model-analysis}}

\begin{proof}[Proof of Proposition~\ref{prop:rule-choice}]
    The result follows from the definition of $\Delta\Phi(\sigma,\pi_0)$.
\end{proof}

\begin{proof}[Proof of Corollary~\ref{crly:rule-choice-convex}]
    If $\psi$ is linear or convex then $\Psi(\sigma,\pi_0)=\E[\psi(\pi_1(Y))]\le\psi(\E[\pi_1(Y)])=\psi(\pi_0)$ by Jensen's inequality.
    Then~\eqref{ineq:rule-choice} holds since the left-hand side is positive and the right-hand side is non-positive.
\end{proof}

\begin{proof}[Proof of Corollary~\ref{crly:rule-choice-wR}]
    The result follows from~\eqref{ineq:rule-choice} and the definition of $\Psi(\sigma,\pi_0)$.
\end{proof}

\begin{proof}[Proof of Lemma~\ref{lem:posterior-comparative-statics}]
    Now
    \[ \parfrac{\pi_1(1)}{\sigma}=\frac{4\pi_0(1-\pi_0)}{(1+4\sigma\pi_0)^2} \]
    and
    \[ \parfrac{\pi_1(1)}{\pi_0}=\frac{1+4\sigma}{(1+4\sigma\pi_0)^2} \]
    are positive.
    In contrast,
    \[ \parfrac{\pi_1(0)}{\sigma}=-\frac{4\pi_0(1-\pi_0)}{(1-4\sigma\pi_0)^2} \]
    is negative, while
    \[ \parfrac{\pi_1(0)}{\pi_0}=\frac{1-4\sigma}{(1-4\sigma\pi_0)^2} \]
    equals zero if $\sigma=0.25$ and is positive otherwise.
\end{proof}

\begin{proof}[Proof of Proposition~\ref{prop:Delta-Phi-comparative-statics-sigma}]
    Now $2\sigma\pi_0$ is increasing in $\sigma$, while $\psi(\pi_0)$ is constant in $\sigma$.
    Moreover, by~\eqref{eq:posterior-mean} and Lemma~\ref{lem:posterior-comparative-statics}, increasing $\sigma$ induces a mean-preserving spread in the distribution of $\pi_1(Y)$.
    This spread weakly increases $\Psi(\sigma,\pi_0)$ when $\psi$ is linear or convex \cite[Theorem~2]{Rothschild-Stiglitz-1970-JET}.
\end{proof}

\begin{proof}[Proof of Proposition~\ref{prop:Delta-Phi-comparative-statics-prior-differentiable}]
    We have
    \begin{align*}
        \parfrac{\Psi(\sigma,\pi_0)}{\pi_0}
        &= \parfrac{}{\pi_0}\left[(0.5+2\sigma\pi_0)\psi(\pi_1(1))+(0.5-2\sigma\pi_0)\psi(\pi_1(0))\right] \\
        &= 2\sigma\psi(\pi_1(1))+(0.5+2\sigma\pi_0)\psi'(\pi_1(1))\parfrac{\pi_1(1)}{\pi_0} \\
        &\quad -2\sigma\psi(\pi_1(0))+(0.5-2\sigma\pi_0)\psi'(\pi_1(0))\parfrac{\pi_1(0)}{\pi_0} \\
        &= 2\sigma\left[\psi(\pi_1(1))-\psi(\pi_1(0))\right]+\E\left[\psi'(\pi_1(Y))\parfrac{\pi_1(Y)}{\pi_0}\right]
    \end{align*}
    and therefore
    \begin{align}
        \parfrac{\Delta\Phi(\sigma,\pi_0)}{\pi_0}
        &= \parfrac{}{\pi_0}\left[2\sigma\pi_0+\Psi(\sigma,\pi_0)-\psi(\pi_0)\right] \notag \\
        &= 2\sigma+2\sigma\left[\psi(\pi_1(1))-\psi(\pi_1(0))\right]+\E\left[\psi'(\pi_1(Y))\parfrac{\pi_1(Y)}{\pi_0}\right]-\psi'(\pi_0). \label{eq:Delta-Phi-derivative-prior}
    \end{align}
    The result follows.
\end{proof}

\begin{proof}[Proof of Proposition~\ref{prop:Delta-Phi-comparative-statics-prior-wR}]
    Substituting $\psi(\pi)=w\,R(\pi)$ into~\eqref{eq:Delta-Phi-derivative-prior} gives
    \begin{equation}
        \parfrac{\Delta\Phi(\sigma,\pi_0)}{\pi_0}
        = 2\sigma+w\left(2\sigma\left[R(\pi_1(1))-R(\pi_1(0))\right]+\E\left[R'(\pi_1(Y))\parfrac{\pi_1(Y)}{\pi_0}\right]-R'(\pi_0)\right). \label{eq:Delta-Phi-derivative-wR}
    \end{equation}
    If~\eqref{ineq:Delta-Phi-comparative-statics-prior-differentiable} does not hold then the bracketed term in~\eqref{eq:Delta-Phi-derivative-wR} must be negative because $\sigma>0$.
    In that case, defining
    \[ w^*\equiv\frac{2\sigma}{R'(\pi_0)-\E\left[R'(\pi_1(Y))\parfrac{\pi_1(Y)}{\pi_0}\right]-2\sigma\left[R(\pi_1(1))-R(\pi_1(0))\right]} \]
    yields the result.
\end{proof}

\subsection{Results in Section~\ref{sec:wage-example}}

\begin{proof}[Proof of Lemma~\ref{lem:wage-example-expected-payoff}]
    Suppose the expert chooses the complex rule.
    If
    \begin{equation}
        \pi^*>\pi_1(1)=\frac{(1+4\sigma)\pi_0}{1+4\sigma\pi_0} \label{eq:wage-example-lower-threshold-origin}
    \end{equation}
    then he always receives a reputation payoff of zero.
    Substituting $\pi_0=\underline\pi$ into~\eqref{eq:wage-example-lower-threshold-origin} and rearranging for $\underline\pi$ yields~\eqref{eq:wage-example-lower-threshold}.
    Conversely, if
    \begin{equation}
        \pi^*\le\pi_1(0)=\frac{(1-4\sigma)\pi_0}{1-4\sigma\pi_0} \label{eq:wage-example-upper-threshold-origin}
    \end{equation}
    then he always receives a reputation payoff of $w$.
    Substituting $\pi_0=\overline\pi$ into~\eqref{eq:wage-example-upper-threshold-origin} and rearranging for $\overline\pi$ yields~\eqref{eq:wage-example-upper-threshold}.
    Finally, if $\pi_1(0)\le\pi^*<\pi_1(1)$ then the expert receives a reputation payoff of $w$ if and only if $Y=1$, which occurs with probability $(0.5+2\sigma\pi_0)$.
    Equation~\eqref{eq:wage-example-expected-payoff-complex} follows from the definition of $\Phi(C,\sigma,\pi_0)$.

    Now suppose the expert chooses the simple rule.
    Then he receives a reputation payoff of $w$ if and only if $\pi_0\ge\pi^*$.
    Equation~\eqref{eq:wage-example-expected-payoff-simple} follows from the definition of $\Phi(S,\sigma,\pi_0)$.
\end{proof}

\begin{proof}[Proof of Proposition~\ref{prop:wage-example-choice}]
    Fix $\sigma$ and consider the thresholds $\underline\pi$ and $\overline\pi$ defined in Lemma~\ref{lem:wage-example-expected-payoff}.
    Now $\underline\pi<\pi^*$ and $\overline\pi>\pi^*$, so if $\pi_0<\pi^*$ then
    \begin{align*}
        \Phi(C,\sigma,\pi_0)-\Phi(S,\sigma,\pi_0)
        &> \Phi(C,\sigma,\pi^*)-\Phi(S,\sigma,\pi_0) \\
        &= (0.5+2\sigma\pi^*)(1+w)-0.5 \\
        &> 0
    \end{align*}
    and if $\pi_0\ge\overline\pi$ then $\Phi(C,\sigma,\pi)-\Phi(S,\sigma,\pi)=2\sigma\pi_0>0$.
    Therefore, if the expert chooses the simple rule then $\pi^*\le\pi_0<\overline\pi$.
    For such values of $\pi_0$ we have $\Phi(C,\sigma,\pi_0)=(0.5+2\sigma\pi_0)(1+w)$ and $\Phi(S,\sigma,\pi_0)=0.5+w$.
    These functions of $\pi_0$ coincide when $\pi_0=\pi^\dagger$ with $\pi^\dagger$ defined as in~\eqref{eq:wage-example-prior-dagger}.
    If $\pi^\dagger\le\pi^*$ then $\Phi(C,\sigma,\pi_0)\ge\Phi(S,\sigma,\pi_0)$ for all $\pi_0\in(0,1)$.
    On the other hand, if $\pi^\dagger>\pi^*$ then $\Phi(C,\sigma,\pi_0)<\Phi(S,\sigma,\pi_0)$ if and only if $\pi^*\le\pi_0<\min\{\pi^\dagger,\overline\pi\}$.
    The result follows.
\end{proof}

\begin{proof}[Proof of Corollary~\ref{crly:wage-example-wage-threshold}]
    Let $\pi^\dagger$ be defined as in~\eqref{eq:wage-example-prior-dagger}.
    If $\pi^*=\pi_0$ then $\overline\pi>\pi_0$ and so the expert chooses the complex rule if and only if $\pi^\dagger\le\pi^*=\pi_0$.
    This occurs if and only if $w\le4\sigma\pi_0/(1-4\sigma\pi_0)$.
\end{proof}

\subsection{Results in Section~\ref{sec:extensions}}

\begin{proof}[Proof of Proposition~\ref{prop:measurement-errors-complex-accuracy}]
    By the Law of Total Probability, we have
    \begin{align*}
        \Pr(A=\hat{X}_\theta)
        &= \Pr(A=1,\hat{X}_\theta=1)+\Pr(A=0,\hat{X}_\theta=0) \\
        &= \Pr(A=1,\hat{X}_\theta=1\,\vert\,X_\theta=1)\Pr(X_\theta=1)+\Pr(A=1,\hat{X}_\theta=1\,\vert\,X_\theta=0)\Pr(X_\theta=0) \\
        &\quad +\Pr(A=0,\hat{X}_\theta=0\,\vert\,X_\theta=1)\Pr(X_\theta=1)+\Pr(A=0,\hat{X}_\theta=0\,\vert\,X_\theta=0)\Pr(X_\theta=0).
    \end{align*}
    Now $A$ and $\hat{X}_\theta$ are conditionally independent given $X_\theta$, so
    \[ \Pr(A,\hat{X}_\theta\,\vert\,X_\theta)=\Pr(A\,\vert\,X_\theta)\Pr(\hat{X}_\theta\,\vert\,X_\theta) \]
    and therefore
    \[ \Pr(A,\hat{X}_\theta\,\vert\,X_\theta)\Pr(X_\theta)=\Pr(A,X_\theta)\Pr(\hat{X}_\theta\,\vert\,X_\theta)\]
    by the definition of $\Pr(A\,\vert\,X_\theta)$.
    Thus
    \begin{align*}
        \Pr(A=\hat{X}_\theta)
        &= \Pr(A=1,X_\theta=1)\Pr(\hat{X}_\theta=1\,\vert\,X_\theta=1)+\Pr(A=1,X_\theta=0)\Pr(\hat{X}_\theta=1\,\vert\,X_\theta=0) \\
        &\quad +\Pr(A=0,X_\theta=1)\Pr(\hat{X}_\theta=0\,\vert\,X_\theta=1)+\Pr(A=0,X_\theta=0)\Pr(\hat{X}_\theta=0\,\vert\,X_\theta=0).
    \end{align*}
    Equation~\eqref{eq:measurement-errors-complex-accuracy} follows from substituting in $\Pr(A,X_\theta)$, $\Pr(\hat{X}_\theta\,\vert\, X_\theta)$, and $a=0.5$.
\end{proof}

\begin{proof}[Proof of Lemma~\ref{lem:known-type-posterior}]
    The advisee first observes the chosen rule $r$ and forms an intermediate belief
    \begin{align*}
        \pi_r
        &\equiv \Pr(\theta=1\,\vert\,r) \\
        &= \frac{\Pr(r\,\vert\,\theta=1)\Pr(\theta=1)}{\Pr(r)} \\
        &= \begin{cases}
            \frac{(1-p_1)\pi_0}{(1-p_1)\pi_0+(1-p_0)(1-\pi_0)} & \text{if}\ r=S \\
            \frac{p_1\pi_0}{p_1\pi_0+p_0(1-\pi_0)} & \text{if}\ r=C.
        \end{cases}
    \end{align*}
    She then observes the outcome $Y$ and uses her updated belief $\pi_r$, rather than her prior $\pi_0$, to form her posterior as in~\eqref{eq:posterior}.
\end{proof}

\begin{proof}[Proof of Lemma~\ref{lem:known-type-expected-payoffs}]
    Equations~\eqref{eq:known-type-expected-payoff-complex} and~\eqref{eq:known-type-expected-payoff-simple} follow from Lemma~\ref{lem:known-type-posterior} and the definition of the expected payoff $\Phi_\theta(r,\sigma,\pi_0,p_0,p_1)$.
    Equation~\eqref{eq:known-type-expected-payoff} then follows from the Law of Total Expectation.
\end{proof}

\begin{proof}[Proof of Proposition~\ref{prop:known-type-pure-equilibrium}]
    If $(p_0,p_1)=(0,1)$ then $\pi_C=1$ and $\pi_S=0$, and so
    \begin{align*}
        \Phi_0(C,\sigma,\pi_0,p_0,p_1)
        &= 0.5+\psi(1) \\
        &> 0.5+\psi(0) \\
        &= \Phi_0(S,\sigma,\pi_0,p_0,p_1)
    \end{align*}
    since $\psi(0)<\psi(1)$ by assumption.
    Thus, an incompetent expert can profitably deviate by choosing the complex rule with non-zero probability.
    Similarly, if $(p_0,p_1)=(1,0)$ then an incompetent expert can profitably deviate by choosing the simple rule with non-zero probability.
\end{proof}

\begin{proof}[Proof of Lemma~\ref{lem:known-type-wage-example-expected-payoffs}]
    Replacing $\pi_0$ with $\pi_r$ in the proof of Lemma~\ref{lem:wage-example-expected-payoff} gives
    \begin{equation}
        \Phi_\theta(C,\sigma,\pi_0,p_0,p_1)
        =\begin{cases}
            0.5+2\sigma\theta & \text{if}\ \pi_C<\underline\pi \\
            0.5+2\sigma\theta+(0.5+2\sigma\theta)w & \text{if}\ \underline\pi\le\pi_C<\overline\pi \\
            0.5+2\sigma\theta+w & \text{if}\ \overline\pi\le\pi_C
        \end{cases} \label{eq:known-type-wage-example-expected-payoff-complex-proof}
    \end{equation}
    and
    \[ \Phi_\theta(S,\sigma,\pi_0,p_0,p_1)=\begin{cases}
        0.5+w & \text{if}\ \pi_S\ge\pi^* \\
        0.5 & \text{otherwise}.
    \end{cases} \]
    If $\pi^*=\pi_0$ then, by the definitions of $\pi_C$, $\underline\pi$, and $\overline\pi$, we have (i)~$\pi_C<\underline\pi$ if and only if $p_1(1+4\sigma)<p_0$ and (ii)~$\pi_C<\overline\pi$ if and only if $p_1(1-4\sigma)<p_0$.
    Substituting (i) and (ii) into~\eqref{eq:known-type-wage-example-expected-payoff-complex-proof} yields \eqref{eq:known-type-wage-example-expected-payoff-complex}.
    Finally, by the definition of $\pi_S$, we have $\pi_S\ge\pi_0$ if and only if $p_0\ge p_1$.
    Equation~\eqref{eq:known-type-wage-example-expected-payoff-simple} follows.
\end{proof}

\begin{proof}[Proof of Proposition~\ref{prop:known-type-wage-example-best-responses}]
    Consider an incompetent expert.
    Suppose he chooses the complex rule with probability $p$.
    Then, by~\eqref{eq:known-type-wage-example-expected-payoff-complex} and~\eqref{eq:known-type-wage-example-expected-payoff-simple}, he receives expected payoff
    \begin{align*}
        \Omega_0(\sigma,\pi_0,p,p_1)
        &= p\Phi_0(C,\sigma,\pi_0,p,p_1)+(1-p)\Omega_0(S,\sigma,\pi_0,p,p_1) \\
        &= \begin{cases}
            0.5+pw & \text{if}\ p\le p_1(1-4\sigma) \\
            0.5+0.5pw & \text{if}\ p_1(1-4\sigma)<p<p_1 \\
            0.5+(1-0.5p)w& \text{if}\ p_1\le p\le p_1(1+4\sigma) \\
            0.5+(1-p)w & \text{if}\ p_1(1+4\sigma)<p,
        \end{cases}
    \end{align*}
    which is uniquely maximized by $p=p_1$.
    Equation~\eqref{eq:known-type-wage-example-best-response-incompetent} follows.

    Now consider a competent expert.
    Suppose he chooses the complex rule with probability $p$.
    For convenience, let
    \begin{align*}
        \Omega(p)
        &\equiv \Omega_1(\sigma,\pi_0,p_0,p) \\
        &= p\Phi_1(C,\sigma,\pi_0,p_0,p)+(1-p)\Phi_1(S,\sigma,\pi_0,p_0,p)
    \end{align*}
    denote his objective given the (exogenous to him) parameters $\sigma$, $\pi_0$, $p_0$, and $w$.
    There are several cases to consider:
    \begin{enumerate}

        \item
        Suppose $4\sigma=1$.
        Then
        \[ \Phi_1(C,\sigma,\pi_0,p_0,p)=\begin{cases}
            1 & \text{if}\ p<\frac{p_0}{2} \\
            1+w & \text{if}\ \frac{p_0}{2}\le p
        \end{cases} \]
        by~\eqref{eq:known-type-wage-example-expected-payoff-complex} and hence
        \begin{align*}
            \Omega(p)
            &= \begin{cases}
                0.5+(0.5-w)p+w & p<\frac{p_0}{2}\\
                0.5+0.5p+w & \frac{p_0}{2}\le p\le p_0 \\
                0.5+0.5p+wp & p_0<p
            \end{cases}
        \end{align*}
        by~\eqref{eq:known-type-wage-example-expected-payoff-simple}.
        But if $p<p_0$ then $\Omega(p)<\Omega(p_0)$, and if $p_0<1$ then $\Omega(p_0)<\Omega(1)$.
        It follows that $\argmax_p\Omega(p)=\{1\}$ in this case.

        \item
        Suppose $4\sigma<1$.
        Then $1-4\sigma>0$, and so
        \begin{align*}
            \Omega(p)
            &= \begin{cases}
                0.5+w+p(2\sigma-w) & p<\frac{p_0}{1+4\sigma} \\
                0.5+w+p(2\sigma(1+w)-0.5w) & \frac{p_0}{1+4\sigma}\le p\le p_0 \\
                0.5+p(2\sigma(1+w)+0.5w) & p_0<p<\frac{p_0}{1-4\sigma} \\
                0.5+p(2\sigma+w) & \frac{p_0}{1-4\sigma}\le p
            \end{cases}
        \end{align*}
        by~\eqref{eq:known-type-wage-example-expected-payoff-complex} and~\eqref{eq:known-type-wage-example-expected-payoff-simple}.
        Now $\Omega(p)$ is increasing in $p$ when $p>p_0$.
        Also, it is increasing, constant, or decreasing in $p$ when $p<p_0/(1+4\sigma)$ and when $p_0/(1+4\sigma)\le p\le p_0$.
        It follows that
        \[ \argmax_p\Omega(p)\subseteq\{0\}\cup\left[\frac{p_0}{1+4\sigma},p_0\right]\cup\{1\}, \]
        and it suffices to compare the values
        \begin{align*}
            \Omega(0)
            &= 0.5+w, \\
            \Omega\left(\frac{p_0}{1+4\sigma}\right)
            &= 0.5+w+\frac{p_0(2\sigma(1+w)-0.5w)}{1+4\sigma}, \\
            \Omega(p_0)
            &= 0.5+w+p_0(2\sigma(1+w)-0.5w), \\
            \intertext{and}
            \Omega(1)
            &= \begin{cases}
                0.5+2\sigma+w & \text{if}\ p_0\le1-4\sigma \\
                0.5+2\sigma(1+w)+0.5w & \text{if}\ 1-4\sigma<p_0.
            \end{cases}
        \end{align*}
        Again, there are two cases:
        \begin{enumerate}

            \item
            Suppose $p_0\le1-4\sigma$.
            Then
            \[ \Omega(1)>\max\left\{\Omega(0),\Omega\left(\frac{p_0}{1+4\sigma}\right),\Omega(p_0)\right\} \]
            since $2\sigma\in(0,0.5)$ and so $2\sigma>\max\{0,2\sigma(1+w)-0.5w\}$.
            Thus $\argmax_p\Omega(p)=\{1\}$ in this case.

            \item
            Suppose $p_0>1-4\sigma$.
            Then $\Omega(1)-\Omega(0)=2\sigma(1+w)-0.5w$ exceeds zero if and only if $2\sigma(1+w)>0.5w$.
            If this condition holds then
            \[ \Omega(1)-\Omega(p)=\left(2\sigma(1+w)-0.5w\right)(1-p) \]
            also exceeds zero for $p\in[p_0/(1+4\sigma),p_0]$.
            The condition always holds when $2\sigma\ge w$.
            Hence, there are three cases:
            \begin{enumerate}

                \item
                If $2\sigma\ge w$ or $2\sigma(1+w)>0.5w$ then $\argmax_p\Omega(p)=\{1\}$.

                \item
                If $2\sigma<w$ and $2\sigma(1+w)<0.5w$, then
                \[ \Omega(0)>\max\left\{\Omega\left(\frac{p_0}{1+4\sigma}\right),\Omega(p_0),\Omega(1)\right\} \]
                and so $\argmax_p\Omega(p)=\{0\}$.

                \item
                If $2\sigma<w$ and $2\sigma(1+w)=0.5$, then $\Omega(p)=0.5+w$ for all $p\in\argmax_p\Omega(p)=\{0\}\cup\left[\frac{p_0}{1+4\sigma},p_0\right]\cup\{1\}$.
                
            \end{enumerate}

        \end{enumerate}

    \end{enumerate}
    Combining these cases and subcases yields
    \begin{align*}
        &\argmax_p\Omega(p) \notag \\
        &= \begin{cases}
            \{0\} & \text{if}\ 4\sigma<1\ \text{and}\ p_0>1-4\sigma\ \text{and}\ 2\sigma<w\ \text{and}\ 2\sigma(1+w)<0.5w \\
            \{0\}\cup[\frac{p_0}{1+4\sigma},p_0]\cup\{1\} & \text{if}\ 4\sigma<1\ \text{and}\ p_0>1-4\sigma\ \text{and}\ 2\sigma<w\ \text{and}\ 2\sigma(1+w)=0.5w \\
            \{1\} & \text{otherwise}.
        \end{cases}
    \end{align*}
    Equation~\eqref{eq:known-type-wage-example-best-response-competent} follows from rewriting $2\sigma(1+w)\le0.5w$ as $w\le4\sigma/(1-4\sigma)$.
\end{proof}

\begin{proof}[Proof of Corollary~\ref{crly:known-type-wage-example-equilibria}]
    Cases~(i) and~(ii) correspond to the first two pieces of~\eqref{eq:known-type-wage-example-best-response-competent}.
    I consider these two cases separately:
    \begin{enumerate}

        \item[(i)]
        Suppose $4\sigma<1$ and $\max\left\{2\sigma,\frac{4\sigma}{1-4\sigma}\right\}<w$.
        Assume towards a contradiction that $(p_0,p_1)$ defines an equilibrium.
        Then $p_0=p_1$ by Proposition~\ref{prop:known-type-wage-example-best-responses}.
        Moreover, by the same proposition, if $p_0\le1-4\sigma$ then $p_1=1>p_0$, whereas if $p_0>1-4\sigma$ then $p_1=0<p_0$.
        Thus $p_1\not=p_0$; a contradiction.
        It follows that no equilibrium exists in this case.

        \item[(ii)]
        Suppose $4\sigma<1$ and $2\sigma<w=\frac{4\sigma}{1-\sigma}$.
        There are two subcases to consider:
        \begin{enumerate}

            \item
            Suppose $p_0\le1-4\sigma$.
            Assume towards a contradiction that $(p_0,p_1)$ defines an equilibrium.
            Then $p_0=p_1$ and $p_1=1$ by Proposition~\ref{prop:known-type-wage-example-best-responses}.
            Hence $p_1>p_0$; a contradiction.

            \item
            Suppose $p_0\in(1-4\sigma,1]$.
            Then, by Proposition~\ref{prop:known-type-wage-example-best-responses}, the pair $(p_0,p_0)$ defines an equilibrium in which both competent and incompetent experts play best responses.
            Moreover, no other equilibria exist because every other pair involves at least one player not best-responding.

        \end{enumerate}
        It follows that $\{(p,p):1-4\sigma<p\le1\}$ defines the set of equilibria in this case.

    \end{enumerate}
    Outside of these two cases, competent experts have a dominant strategy to play $p_1=1$, to which incompetent experts best-respond by playing $p_0=1$.
    Then $(p_0,p_1)=(1,1)$ defines the unique equilibrium.
\end{proof}

\end{document}